\documentclass[a4paper,fleqn]{cas-dc}

\usepackage[numbers]{natbib}
\usepackage{subcaption}

\usepackage{enumitem}
\usepackage{flushend}

\newtheorem{theorem}{Theorem}

\newtheorem{proposition}{Proposition}
\newdefinition{remark}{Remark}
\newdefinition{definition}{Definition}
\newproof{proof}{Proof}
\newdefinition{assumption}{Assumption}

\definecolor{myred}{rgb}{0.7,0.1,0.16}
\definecolor{myblue}{rgb}{0,0.32,0.7}
\definecolor{mygreen}{rgb}{0.133,0.545,0.133}

\graphicspath{{figures/}}

\begin{document}
\let\WriteBookmarks\relax
\def\floatpagepagefraction{1}
\def\textpagefraction{.001}

\title[mode = title]{Finite Boundary-Layer Residence Certificates for Non-Strict Control Barrier Functions}  
\shorttitle{Finite Boundary-Layer Residence Certificates for Non-Strict Control Barrier Functions}
\shortauthors{T. Han, G. Wang, and B. Wang}

\author[1]{Tianyu Han}
\ead{tianyu.han37@stu-mail.ccny.cuny.edu}
\credit{Methodology, Formal analysis, Investigation, Software, Validation, Visualization, Writing -- original draft}
            
\author[2]{Guangwei Wang}
\ead{gwwang@gzu.edu.cn}
\credit{Methodology, Formal analysis, Writing -- review and editing}

\author[1]{Bo Wang}\cormark[1]
\ead{bwang1@ccny.cuny.edu}
\credit{Conceptualization, Methodology, Formal analysis, Supervision, Project administration, Writing -- review and editing}

\affiliation[1]{organization={Department of Mechanical Engineering, The City College of New York, The City University of New York},
            city={New York},
            state={NY},
            postcode={10031},
            country={USA}}

\affiliation[2]{organization={School of Mechanical Engineering, Guizhou University},
            city={Guiyang},
            state={Guizhou},
            postcode={550025},
            country={China}}

\cortext[1]{Corresponding author.}

\nonumnote{\textit{E-mail addresses:} {tianyu.han37@stu-mail.ccny.cuny.edu} (T. Han), {gwwang@gzu.edu.cn} (G. Wang), {bwang1@ccny.cuny.edu} (B. Wang).}

\begin{abstract}
Non-strict control barrier function (CBF) conditions guarantee safety through forward invariance, but they do not preclude trajectories from remaining near the safe-set boundary for extended continuous time intervals. This paper develops a finite boundary-layer residence certificate for such settings. The certificate preserves the standard non-strict CBF safety condition and uses a bounded auxiliary function whose derivative is bounded away from zero in a prescribed boundary layer, yielding an explicit upper bound on every uninterrupted residence interval. For control-affine systems, a selected auxiliary branch is implemented as an additional affine constraint in a CBF-QP, and a tangential-input compatibility condition is given to ensure simultaneous feasibility with the hard CBF constraint for unconstrained inputs. A local-chart version handles angular or multi-valued auxiliary functions such as \(\operatorname{atan2}\). Single-integrator, double-integrator, and nonholonomic unicycle examples illustrate the resulting radial--tangential construction and its local-chart and feasibility limitations.
\end{abstract}


\maketitle

\section{Introduction}\label{sec:introduction}

Control barrier functions (CBFs) provide a systematic framework for enforcing
safety through forward invariance of admissible sets
\cite{ames2017control}. By converting safety requirements into
derivative inequalities, CBFs have become a standard tool for safety-critical
controller synthesis, especially when combined with quadratic programs (QPs) \cite{xu2015robustness,JANKOVIC2018robust,alan2023control}.
The resulting CBF-based QP controllers have been applied to mobile robots, multi-agent
systems, robotic manipulation, and connected autonomous vehicles
\cite{Han2024Safety,Wang2025Further,jankovic2024multiagent}. In these
methods, the safety constraint is usually imposed as a hard affine constraint,
while stabilization or performance objectives are enforced through a nominal
controller or a relaxed control Lyapunov function (CLF) constraint.

The standard CBF condition is \textit{non-strict} \cite{ames2019control,krstic2023inverse}. 
It guarantees that trajectories cannot leave the safe set, but it does not prevent them from remaining near the safety boundary for a long continuous time. 
This distinction is important in obstacle-avoidance and robotic navigation problems. 
A robot may remain safe, in the sense of avoiding collision, while the active safety constraint blocks the nominal stabilizing direction and prevents task progress. 
Such behavior is often described as boundary sticking, deadlock, or loss of liveness \cite{jankovic2024multiagent,reis2020control,grover2020does}. Thus, forward invariance alone does not rule out prolonged near-boundary behavior.

Recent work has analyzed this issue through the closed-loop dynamics induced by
CBF-QP and CLF-CBF-QP controllers. It has been shown that the interaction
between hard CBF constraints and relaxed CLF objectives can generate undesired
equilibria near the safe-set boundary, with explicit conditions under which
such equilibria are asymptotically stable \cite{reis2020control}. Related
results characterize the existence and location of CBF-QP-induced equilibria
for nonlinear control-affine systems and propose modified QP formulations to
remove certain undesired equilibria while preserving forward invariance
\cite{Tan2024On}. In multi-agent collision avoidance, CBF-based controllers
have also been studied from the viewpoint of safety--liveness tradeoffs,
gridlock, and the stability of boundary equilibria
\cite{jankovic2024multiagent,grover2020does}. More recent analyses of
CBF-based safety filters further show that undesirable closed-loop phenomena
may include not only undesired equilibria but also unbounded trajectories and
limit cycles \cite{chen2024equilibria,mestres2026control}. These studies
clarify how particular boundary equilibria and related closed-loop pathologies
arise, and how some of them can be modified or removed.

What remains missing is a certificate for the duration of near-boundary motion itself. 
Equilibrium-based analyses can identify where boundary pathologies arise and, in some cases, modify the controller to remove particular undesired equilibria. 
However, the absence or instability of such equilibria does not rule out long transients, sliding-like motion along active safety constraints, recurrent near-boundary motion, or repeated conflicts between safety and performance objectives. 
These are trajectory-level phenomena. 
The relevant question is therefore not only whether a closed loop has an undesirable boundary equilibrium, but whether each continuous stay near the safe-set boundary can be explicitly bounded in time. 
The missing ingredient is a \textit{residence-time certificate}: a condition that works with the standard non-strict CBF inequality and nevertheless bounds every uninterrupted near-boundary residence interval.

A natural way to attack boundary residence is to make the barrier inequality
\textit{strict}, thereby forcing trajectories away from the boundary
\cite{wang2026universal}. 
This changes the safety mechanism itself. 
For nonholonomic or underactuated systems, such strictness may be structurally
unavailable at configurations where the input has no outward effect on the
barrier derivative. 
It is also awkward in QP-based safety filters: strict inequalities do not define closed feasible sets, and fixed positive margins can destroy feasibility or exclude otherwise safe controls. 
Thus, strict barrier conditions are useful when available, but they are not a general remedy for boundary-layer residence in standard CBF-QP implementations.

This paper instead keeps the usual non-strict CBF condition as the safety
constraint and adds an auxiliary condition only to certify residence time. 
The construction is inspired by Matrosov-type reasoning, where a non-strict
primary Lyapunov inequality is supplemented by auxiliary trajectory information
\cite{loria2005nested}. 
Our use is finite-interval rather than asymptotic. 
We choose an auxiliary function that is bounded on the relevant domain or local chart and impose, inside the prescribed boundary layer, a one-sided derivative bound. 
The auxiliary function therefore cannot evolve monotonically inside the layer for arbitrarily long, which yields an explicit upper bound on each uninterrupted boundary-layer residence interval.

The contributions of this paper are fourfold. 
First, we formulate a trajectory-level certificate for finite continuous residence in a prescribed boundary layer under the standard non-strict CBF condition. 
The certificate preserves forward invariance and gives an explicit upper bound on every uninterrupted residence interval through a bounded auxiliary function with a one-sided derivative bound. 
Second, we develop a local-chart version for angular or multi-valued auxiliary functions, such as $\operatorname{atan2}$, and specify the associated branch-cut and singularity limitations. 
Third, for control-affine systems, we enforce a selected auxiliary branch as an additional affine constraint in a CBF-QP. 
Fourth, we give a tangential-input compatibility condition ensuring that this auxiliary constraint can coexist with the hard CBF constraint for unconstrained inputs. 
Single-integrator, double-integrator, and unicycle examples demonstrate the radial--tangential construction, its QP implementation, and its local-chart and feasibility limitations.

\section{Safety and Boundary-Layer Residence}\label{sec:formulation}

\textit{Notation.} The Euclidean norm on $\mathbb{R}^n$ is denoted by $|\cdot|$. For any set $S\subset\mathbb{R}^n$, $\partial S$ and $\operatorname{Int}(S)$ denote its boundary and interior, respectively. The class $\mathcal{K}$ consists of continuous functions $\alpha:\mathbb{R}_{\ge0}\to\mathbb{R}_{\ge0}$ such that $\alpha(0)=0$ and $\alpha$ is strictly increasing. The class $\mathcal{K}_\infty$ consists of those functions in $\mathcal{K}$ that are unbounded. A continuous function $\alpha:\mathbb{R}\to\mathbb{R}$ is an extended class $\mathcal{K}_\infty$ function, denoted by $\alpha\in\mathcal{K}^e_\infty$, if $\alpha(0)=0$, $\alpha$ is strictly increasing, and $\lim_{s\to\pm\infty}\alpha(s)=\pm\infty$.

Consider the nonlinear time-varying system
\begin{equation}\label{eq:tv-system}
    \dot{x}=f(t,x),
\end{equation}
where $f:\mathbb{R}_{\ge0}\times\mathbb{R}^n\to\mathbb{R}^n$ is continuous in $t$ and locally Lipschitz in $x$. Let a continuously differentiable function $h:\mathbb{R}_{\ge0}\times\mathbb{R}^n\to\mathbb{R}$ define the time-varying \textit{safe set}
\begin{equation}
    \mathcal{C}(t):=\{x\in\mathbb{R}^n:h(t,x)\ge0\}.
    \label{eq:tv-safeset}
\end{equation}
The set $\mathcal{C}$ is \emph{forward invariant} for \eqref{eq:tv-system} if, for every initial time $t_0\ge0$ and every initial condition $x_0\in\mathcal{C}(t_0)$, the corresponding solution satisfies $x(t;t_0,x_0)\in\mathcal{C}(t)$ for all $t\ge t_0$ for which the solution exists. The system \eqref{eq:tv-system} is \emph{safe} with respect to $\mathcal{C}(\cdot)$ if $\mathcal{C}(\cdot)$ is forward invariant.

A standard way to verify safety is to certify forward invariance of the set $\mathcal C$ through a barrier function \cite{ames2019control}. 
\begin{definition}[Barrier function]\label{def:non-strict-barrier}\rm
A continuously differentiable function $h:\mathbb R_{\ge 0}\times\mathbb R^n\to\mathbb R$ is called a
\emph{non-strict barrier function} for \eqref{eq:tv-system} with respect to
$\mathcal C$ if zero is a regular value\footnote{Zero is a regular value
of $h(t,\cdot)$ if, for every $t$ under consideration and every
$x\in\partial\mathcal C(t)$, one has $\nabla_x h(t,x)\neq 0$.}
of $h(t,\cdot)$ and there exists $\alpha_h\in\mathcal K^e_\infty$ such that
the derivative of $h$ along \eqref{eq:tv-system} satisfies, for all
$t\ge 0$ and all $x\in\mathcal C(t)$,
\begin{equation}
    \dot h(t,x)
    :=
    \partial_t h(t,x)+L_fh(t,x)
    \ge
    -\alpha_h\big(h(t,x)\big).
    \label{eq:ZBF-tv}
\end{equation}
If the inequality in \eqref{eq:ZBF-tv} is strict at every boundary point, that is,
$\dot h(t,x)>0$ for all $t\ge 0$ and all $x\in\partial\mathcal C(t)$, then
$h$ is called a \emph{strict barrier function}.
\end{definition}

For a fixed thickness $\rho>0$, define the \textit{boundary layer}
\begin{equation}
    \Sigma_\rho(t)
    :=
    \{x\in\mathbb{R}^n:0\le h(t,x)\le\rho\}.
    \label{eq:boundary-layer-main}
\end{equation}
A compact interval $[t_a,t_b]$ is called an uninterrupted residence interval in
$\Sigma_\rho(\cdot)$ for a solution $x(\cdot)$ if
$x(t)\in\Sigma_\rho(t)$ for all $t\in[t_a,t_b]$.
Boundary residence therefore refers to residence in the layer
$\Sigma_\rho(t)$, not necessarily on the zero level set of $h$.
The non-strict barrier condition \eqref{eq:ZBF-tv} is a safety condition.
By the standard comparison argument, it guarantees forward invariance of
$\mathcal C(\cdot)$. However, it does not preclude uninterrupted residence
intervals of arbitrarily large length in $\Sigma_\rho(\cdot)$. Thus, although
\eqref{eq:ZBF-tv} certifies safety, it does not certify finite residence time
near the safe-set boundary.

A strict barrier condition can force trajectories away from the boundary, but
such strictness may be structurally infeasible for nonholonomic or underactuated
systems when the input cannot affect $\dot h$ in certain boundary
configurations. Strict inequalities are also inconvenient in QP-based safety
filters, which use closed affine constraints. This motivates the approach below:
we retain the standard non-strict CBF condition for safety and add an auxiliary
condition only to certify finite continuous residence time in
$\Sigma_\rho(\cdot)$.

\section{Residence Certification and QP Implementation}\label{sec:main-results}
\subsection{Finite Continuous Residence-Time Certificate}\label{subsec:certificate}

Inspired by Matrosov-type conditions~\cite{loria2005nested,ROUHABLAL}, we
augment the non-strict CBF condition with a bounded auxiliary function. In
Lyapunov theory, a non-strict inequality $\dot V(x)\le0$ may fail to imply
convergence because trajectories can remain where $\dot V(x)=0$.
Matrosov-type arguments address this by using auxiliary functions whose
derivatives reveal motion on this zero-decrease set and exclude nontrivial
invariant behavior.
Here the role is analogous, but the objective is finite residence time rather
than asymptotic convergence. The non-strict CBF condition guarantees forward
invariance but permits tangency and prolonged motion near the boundary. We
therefore require a bounded auxiliary function to have a derivative bounded
away from zero inside the prescribed boundary layer. Since such a function
cannot evolve monotonically for arbitrarily long while remaining bounded, this
condition yields a finite upper bound on every uninterrupted residence interval.

\begin{theorem}\label{thm:main}
Consider \eqref{eq:tv-system} and the safe set \eqref{eq:tv-safeset}. 
Let $h$ be a non-strict barrier function satisfying \eqref{eq:ZBF-tv}.
Fix $\rho>0$ and let $\Sigma_\rho(t)$ be defined by
\eqref{eq:boundary-layer-main}. Suppose that there exist a
continuously differentiable function
$W:\mathbb R_{\ge0}\times\mathbb R^n\to\mathbb R$, two constants $M>0$ and $\eta_\rho>0$, such that
\begin{equation}
    |W(t,x)|\le M,
    \quad
    \forall t\ge0,\quad
    \forall x\in\mathcal C(t),
    \label{eq:bounded-W-main}
\end{equation}
and
\begin{equation}
    |\dot W(t,x)|\ge \eta_\rho,
    \quad
    \forall t\ge0,\quad
    \forall x\in\Sigma_\rho(t).
    \label{eq:aux-derivative-main}
\end{equation}
Then, for every solution $x(\cdot)$ with $x(t_0)\in\mathcal C(t_0)$, the safe set
$\mathcal C(\cdot)$ is forward invariant on the interval of existence of the
solution. Moreover, if $[t_a,t_b]$ is an uninterrupted residence interval
contained in the interval of existence, with $t_a\ge t_0$, and
\begin{equation}
    x(t)\in\Sigma_\rho(t),
    \quad
    \forall t\in[t_a,t_b],
    \label{eq:residence-interval-condition}
\end{equation}
then
\begin{equation}
    t_b-t_a\le \frac{2M}{\eta_\rho}.
    \label{eq:10}
\end{equation}
Consequently, every uninterrupted residence interval in $\Sigma_\rho(t)$ is
finite, and eventual continuous residence in $\Sigma_\rho(t)$ is excluded.
\end{theorem}

\begin{proof}
Since $h$ is a non-strict barrier function satisfying \eqref{eq:ZBF-tv},
the standard comparison argument implies that $\mathcal C(\cdot)$ is forward
invariant on the interval of existence of the solution.

Fix an uninterrupted residence interval $[t_a,t_b]$ satisfying
\eqref{eq:residence-interval-condition}. Then \eqref{eq:aux-derivative-main}
gives $|\dot W(t,x(t))|\ge \eta_\rho$ for all $t\in[t_a,t_b]$. Since
$t\mapsto \dot W(t,x(t))$ is continuous on $[t_a,t_b]$, it cannot change sign
on this interval; otherwise it would vanish at some time. Hence
$\dot W(t,x(t))$ is either at least $\eta_\rho$ or at most $-\eta_\rho$
throughout $[t_a,t_b]$.

Let $\Delta W:=W(t_b,x(t_b))-W(t_a,x(t_a))$. From the constant-sign property,
\begin{equation}
    |\Delta W|
    =
    \left|\int_{t_a}^{t_b}\dot W(t,x(t))\,dt\right|
    \ge
    \eta_\rho(t_b-t_a).
    \label{eq:W-variation-lower}
\end{equation}
On the other hand, since $x(t_a),x(t_b)\in\mathcal C(t)$ at the corresponding
times, \eqref{eq:bounded-W-main} gives
\begin{equation}
    |\Delta W|
    \le
    |W(t_b,x(t_b))|+|W(t_a,x(t_a))|
    \le
    2M.
    \label{eq:W-variation-upper}
\end{equation}
Combining \eqref{eq:W-variation-lower} and \eqref{eq:W-variation-upper}
yields $t_b-t_a\le 2M/\eta_\rho$.

Finally, if eventual continuous residence in $\Sigma_\rho(t)$ occurred, then
there would exist uninterrupted residence intervals of arbitrarily large
length, contradicting the bound above. This proves the result.\qed
\end{proof}

\begin{remark}
The boundedness assumption on $W$ can be localized. If a solution is known to
remain in a positively invariant subset $\mathcal D(t)\subset\mathcal C(t)$,
then it is enough to require $|W(t,x)|\le M$ on $\mathcal D(t)$, and the same
proof applies to residence intervals contained in $\mathcal D(t)$. This
localization principle is used below in the local-chart certificate, where the
auxiliary function may be single-valued and bounded only on a selected chart.
\end{remark}

\begin{remark}
The constant $\eta_\rho$ is fixed after the boundary-layer thickness $\rho$ is
chosen, but it may depend on $\rho$. Theorem~\ref{thm:main} bounds each
uninterrupted residence interval in $\Sigma_\rho(t)$ separately; it does not
bound cumulative residence time or the number of exits and re-entries. Additional
nominal or Lyapunov-based control objectives may reduce repeated re-entries in
practice, but this cumulative behavior is outside the scope of the certificate.
\end{remark}

The localization idea is useful for angular auxiliary functions, which may not
be globally single-valued. Once a chart is selected, the corresponding branch
can be treated as an ordinary continuously differentiable auxiliary function on
that chart.

\begin{proposition}\label{prop:local-chart}
Consider \eqref{eq:tv-system} and the safe set \eqref{eq:tv-safeset}. 
Let $h$ be a non-strict barrier function satisfying \eqref{eq:ZBF-tv}. 
Fix $\rho>0$ and let $\Sigma_\rho(t)$ be defined by
\eqref{eq:boundary-layer-main}. 
Let $U$ be a chart domain on which $W_U(t,x)$ is single-valued and continuously
differentiable. Suppose that there exist constants $M_U>0$ and
$\eta_{\rho,U}>0$ such that
\begin{equation}
    |W_U(t,x)|\le M_U,
    \qquad
    \forall t\ge0,\ \forall x\in\mathcal C(t)\cap U,
    \label{eq:local-chart-bounded}
\end{equation}
and
\begin{equation}
    |\dot W_U(t,x)|\ge \eta_{\rho,U},
    \qquad
    \forall t\ge0,\ \forall x\in\Sigma_\rho(t)\cap U.
    \label{eq:local-chart-derivative}
\end{equation}
Then $\mathcal C(\cdot)$ is forward invariant. Moreover, every uninterrupted
residence interval $[t_a,t_b]$ satisfying
$x(t)\in\Sigma_\rho(t)\cap U$ for all $t\in[t_a,t_b]$ obeys
\begin{equation}
    t_b-t_a \le \frac{2M_U}{\eta_{\rho,U}}.
    \label{eq:local-chart-bound}
\end{equation}
\end{proposition}

\begin{proof}
Since $h$ satisfies \eqref{eq:ZBF-tv}, the standard comparison argument for
barrier functions implies that $\mathcal C(\cdot)$ is forward invariant on the
interval of existence of every solution starting from $\mathcal C(t_0)$.
Now fix an interval $[t_a,t_b]$ satisfying the stated local-chart residence
conditions. Since $x(t)\in U$ on this interval, $W_U(t,x(t))$ is well-defined
and continuously differentiable along the trajectory. By
\eqref{eq:local-chart-derivative}, $|\dot W_U(t,x(t))|\ge\eta_{\rho,U}$ for all
$t\in[t_a,t_b]$. Continuity of $t\mapsto \dot W_U(t,x(t))$ implies that
$\dot W_U(t,x(t))$ cannot change sign on $[t_a,t_b]$. Hence
\begin{equation}
    \left|
    W_U(t_b,x(t_b))-W_U(t_a,x(t_a))
    \right|
    \ge
    \eta_{\rho,U}(t_b-t_a).
    \label{eq:local-W-lower}
\end{equation}
On the other hand, $x(t_a),x(t_b)\in\mathcal C(t)\cap U$ at the corresponding
times, so \eqref{eq:local-chart-bounded} gives
\begin{equation}
    \left|
    W_U(t_b,x(t_b))-W_U(t_a,x(t_a))
    \right|
    \le
    2M_U.
    \label{eq:local-W-upper}
\end{equation}
Combining \eqref{eq:local-W-lower} and \eqref{eq:local-W-upper} gives
\eqref{eq:local-chart-bound}.\qed
\end{proof}

Proposition~\ref{prop:local-chart} applies to angular coordinates such as
$\operatorname{atan2}$ after a branch is selected. It does not exclude residence
intervals that cross a branch cut or pass through multiple charts. A global
conclusion requires either a globally single-valued auxiliary function or an
additional chart-transition argument.

\subsection{Auxiliary-Augmented CBF-QP Implementation}
\label{subsec:control-affine-qp}

We next show how a selected auxiliary branch can be enforced in a
control-affine system. The purpose of this subsection is not to guarantee QP
feasibility in general; feasibility is addressed separately in
Section~\ref{subsec:compatibility}. Rather, the result below shows that whenever
the hard CBF constraint and the selected auxiliary constraint are feasible along
the closed-loop trajectory, the residence-time certificate of
Section~\ref{subsec:certificate} is inherited by the QP implementation.

Consider the control-affine system
\begin{equation}
    \dot x=f(t,x)+g(t,x)u,
    \label{eq:control-affine-system}
\end{equation}
where $x\in\mathbb R^n$ and $u\in\mathbb R^m$. 
For the function $h$, define
\begin{align*}
    a_h(t,x)&:=\partial_t h(t,x)+L_fh(t,x)+\alpha_h\big(h(t,x)\big),\\
    b_h(t,x)&:=L_gh(t,x)^\top.
\end{align*}
Then the non-strict CBF condition is equivalent to
\begin{equation}
    a_h(t,x)+b_h(t,x)^\top u\ge 0.
    \label{eq:affine-cbf-constraint}
\end{equation}
Let $W$ be an auxiliary function and define
\begin{align*}
    a_W(t,x)&:=\partial_t W(t,x)+L_fW(t,x),\\
    b_W(t,x)&:=L_gW(t,x)^\top.
\end{align*}
Fix a branch parameter $s\in\{1,-1\}$. The choice $s=1$ corresponds to the
positive branch $\dot W\ge\eta_\rho$, while $s=-1$ corresponds to the negative
branch $\dot W\le-\eta_\rho$. Both cases can be written as
\begin{equation}
    s\big(a_W(t,x)+b_W(t,x)^\top u\big)\ge \eta_\rho.
    \label{eq:selected-aux-condition}
\end{equation}
Equivalently, define
\begin{align*}
    a_s(t,x)&:=s\,a_W(t,x)-\eta_\rho,\\
    b_s(t,x)&:=s\,b_W(t,x).
\end{align*}
Then the selected auxiliary branch is the affine constraint
\begin{equation}
    a_s(t,x)+b_s(t,x)^\top u\ge 0.
    \label{eq:selected-aux-branch}
\end{equation}

Given a nominal input $u_{\rm nom}(t,x)$, the auxiliary-augmented CBF-QP is
defined pointwise by
\begin{equation}
\begin{aligned}
    k_s(t,x)&\in\arg\min_{u\in\mathbb R^m}\quad
    \frac12|u-u_{\rm nom}(t,x)|^2\\
    \text{s.t.}\quad
    &a_h(t,x)+b_h(t,x)^\top u\ge 0,\\
    &a_s(t,x)+b_s(t,x)^\top u\ge 0,
    \quad \text{if } x\in\Sigma_\rho(t).
\end{aligned}
\label{eq:auxiliary-cbf-qp}
\end{equation}
Outside $\Sigma_\rho(t)$, the auxiliary constraint is not needed for the
residence certificate. In implementation, it may be imposed on a slightly
larger layer, or through a smooth gate, to avoid abrupt activation. The
certificate only requires that the selected auxiliary inequality hold with a
positive margin whenever the trajectory is in $\Sigma_\rho(t)$.

\begin{proposition}
\label{prop:qp-enforced-certificate}
Consider \eqref{eq:control-affine-system} and the safe set
\eqref{eq:tv-safeset}. Fix $\rho>0$ and let $\Sigma_\rho(t)$ be defined by
\eqref{eq:boundary-layer-main}. Let $k_s$ be a feedback generated by
\eqref{eq:auxiliary-cbf-qp} with a fixed branch $s\in\{1,-1\}$. Suppose that
the QP is feasible along the closed-loop trajectory, that the corresponding solution exists\footnote{The same conclusion holds for Carath\'eodory solutions if $W(t,x(t))$ is
absolutely continuous and the QP inequalities hold almost everywhere along the
solution.}, and that $W$ is bounded by $M$ on $\mathcal C(t)$.
Then $\mathcal C(\cdot)$ is forward invariant. Moreover, every uninterrupted
residence interval in $\Sigma_\rho(t)$ has length at most $2M/\eta_\rho$.
\end{proposition}

\begin{proof}
Along the closed-loop system $u=k_s(t,x)$, the first QP constraint gives
$\dot h(t,x,k_s(t,x))\ge-\alpha_h(h(t,x))$, and hence $\mathcal C(\cdot)$ is
forward invariant by the comparison argument. On any uninterrupted residence
interval in $\Sigma_\rho(t)$, the second QP constraint gives
$s\dot W(t,x(t),k_s(t,x(t)))\ge\eta_\rho$. Thus either
$\dot W\ge\eta_\rho$ on the whole interval or $\dot W\le-\eta_\rho$ on the whole
interval. Since $W$ is bounded by $M$ on $\mathcal C(t)$, the boundedness
argument of Theorem~\ref{thm:main} gives $t_b-t_a\le 2M/\eta_\rho$. The
Carath\'eodory case follows from the same integral argument.
\qed
\end{proof}

\vspace{-0.2cm}
Performance constraints, such as CLF constraints, may be added to
\eqref{eq:auxiliary-cbf-qp} and relaxed by slack variables. The CBF constraint remains
hard for safety, and the auxiliary constraint must remain hard with a positive
effective margin whenever the residence-time certificate is invoked. The branch
parameter $s$ is fixed on each certified residence interval; branch switching is
allowed only if a one-sided auxiliary derivative bound is preserved throughout
the interval.

\subsection{Compatibility of the CBF and Auxiliary Constraints}
\label{subsec:compatibility}

The auxiliary-augmented QP contains two affine constraints: the CBF
constraint for safety and the auxiliary constraint for residence certification.
These constraints are not automatically compatible. This subsection gives a
sufficient input-space algebraic condition under which the auxiliary constraint can be
satisfied without reducing the CBF margin.

\begin{proposition}
\label{prop:tangential-compatibility}
Consider the affine constraints
\begin{equation}
    a_h(t,x)+b_h(t,x)^\top u\ge 0,
    \quad
    a_s(t,x)+b_s(t,x)^\top u\ge 0.
    \label{eq:two-affine-constraints}
\end{equation}
Let $\mathcal N_\rho\subset\mathbb R_{\ge0}\times\mathbb R^n$ be a region
containing the prescribed boundary layer. Suppose that, for every
$(t,x)\in\mathcal N_\rho$, there exists $u_h(t,x)\in\mathbb R^m$ such that
\begin{equation}
    a_h(t,x)+b_h(t,x)^\top u_h(t,x)\ge 0.
    \label{eq:cbf-feasible-point}
\end{equation}
Suppose further that, for every $(t,x)\in\mathcal N_\rho$, there exists a
direction $d(t,x)\in\mathbb R^m$ and a constant $\kappa>0$ such that
\begin{equation}
    b_h(t,x)^\top d(t,x)=0,
    \quad
    b_s(t,x)^\top d(t,x)\ge \kappa.
    \label{eq:compatibility-direction}
\end{equation}
Then the two constraints in \eqref{eq:two-affine-constraints} are
simultaneously feasible for unconstrained inputs $u\in\mathbb R^m$ at every
$(t,x)\in\mathcal N_\rho$.
\end{proposition}

\begin{proof}
Fix $(t,x)\in\mathcal N_\rho$ and set $u=u_h+\lambda d$ with $\lambda\ge0$.
Because $b_h^\top d=0$, the CBF left-hand side is unchanged:
$a_h+b_h^\top u=a_h+b_h^\top u_h\ge0$. The auxiliary left-hand side satisfies
$a_s+b_s^\top u=a_s+b_s^\top u_h+\lambda b_s^\top d
\ge a_s+b_s^\top u_h+\lambda\kappa$. Hence the auxiliary constraint holds for
any
\begin{equation}
    \lambda
    \ge
    \max\left\{
    0,\,
    -\frac{a_s(t,x)+b_s(t,x)^\top u_h(t,x)}{\kappa}
    \right\}.
\end{equation}
Thus both affine constraints are feasible at $(t,x)$. \qed
\end{proof}

The direction $d(t,x)$ is CBF-neutral: it leaves the CBF constraint unchanged
while increasing the selected auxiliary constraint. For a single point, the
weaker condition $b_s(t,x)^\top d(t,x)>0$ is sufficient; the uniform lower bound
$\kappa>0$ gives a regional compatibility condition on $\mathcal N_\rho$.

In planar obstacle-avoidance problems, the condition \eqref{eq:compatibility-direction} has a radial--tangential interpretation. The CBF constraint regulates radial motion relative to the
obstacle, while the auxiliary constraint can be enforced through a tangential
input direction. This is the mechanism used in the examples of the next section.
With actuator limits or additional hard constraints, Proposition~\ref{prop:tangential-compatibility}
is no longer sufficient. For an admissible input set $\Omega_u(t,x)$, feasibility
must be checked by intersecting $\Omega_u(t,x)$ with the CBF and auxiliary
half-spaces. If the auxiliary constraint is softened, the CBF constraint should
remain hard, and the residence certificate is retained only when the effective
auxiliary margin remains positive.

\section{Illustrative Examples}\label{sec:examples}

The following examples illustrate how the auxiliary function can be chosen in
representative obstacle-avoidance problems. A universal construction of the auxiliary function $W$ is
not claimed; the auxiliary function and compatible input direction are system-
and geometry-dependent. In the planar examples below, the construction follows
a radial--tangential pattern: the CBF regulates radial safety relative to the
obstacle, while the auxiliary constraint induces tangential motion near the
boundary layer.
We consider a single integrator, a double integrator, and a nonholonomic
unicycle. The simulations illustrate trajectory-level residence behavior and
should be interpreted together with the stated local-chart and feasibility
conditions.

\subsection{Single-Integrator System}

Consider the planar single-integrator system
\begin{equation}
    \dot{\mathbf x}=\mathbf u,
    \label{eq:single-integrator}
\end{equation}
where $\mathbf x=(x_1,x_2)\in\mathbb R^2$ and $\mathbf u\in\mathbb R^2$. Let the
obstacle be the disk centered at $\mathbf c=(0,3)$ with radius $R=1.5$.
The safe set is
\begin{equation}
    \mathcal C
    =
    \{\mathbf x:h(\mathbf x)=|\mathbf x-\mathbf c|^2-R^2\ge0\}.
    \label{eq:safe_set_circle}
\end{equation}
We choose $\alpha_h(s)=2s$. Since $f=0$ and $g=I_2$, the CBF coefficients in
Section~\ref{subsec:control-affine-qp} are
\begin{align*}
    a_h(\mathbf x)=\alpha_h(h(\mathbf x)),\quad
    b_h(\mathbf x)=\nabla h(\mathbf x)=2(\mathbf x-\mathbf c).
\end{align*}
Thus the non-strict CBF constraint is
$a_h(\mathbf x)+b_h(\mathbf x)^\top\mathbf u\ge0$.

The auxiliary function is the angular coordinate around the obstacle,
$W(\mathbf x):=\operatorname{atan2}(x_2-c_2,x_1-c_1)$, interpreted on a selected
local branch. On this branch, $W$ is single-valued and continuously
differentiable away from $\mathbf x=\mathbf c$, and
\begin{equation}
    \dot W(\mathbf x,\mathbf u)
    =
    \frac{(x_1-c_1)u_2-(x_2-c_2)u_1}{|\mathbf x-\mathbf c|^2}.
    \label{eq:single-Wdot}
\end{equation}
Equivalently, with $a_W(\mathbf x)=0$ and $b_W(\mathbf x)=\nabla W(\mathbf x)$,
one has $\dot W(\mathbf x,\mathbf u)=a_W(\mathbf x)+b_W(\mathbf x)^\top
\mathbf u$.
We select the negative auxiliary branch and localize its effect near the safety
boundary by a smooth gating function. Let $\eta>0$ and $h_{\rm gate}>0$, and
define $\sigma(h):={h_{\rm gate}^2}/{(h_{\rm gate}^2+h^2)}$.
Then $\sigma(0)=1$, and $\sigma(h)$ decreases as $|h|$ increases. Hence the
auxiliary action is strongest near the boundary $h=0$ and weakens away from it.
The selected auxiliary condition is
$\dot W(\mathbf x,\mathbf u)\le-\eta\sigma(h(\mathbf x))$. In the notation of
Section~\ref{subsec:control-affine-qp}, this is the constraint
$a_s(\mathbf x)+b_s(\mathbf x)^\top\mathbf u\ge0$, where
\begin{align*}
    a_s(\mathbf x)=-\eta\sigma(h(\mathbf x)),\quad
    b_s(\mathbf x)=-b_W(\mathbf x).
\end{align*}
Thus, the gating function enters the QP through the scalar term $a_s$.

For controller implementation, we use a $\gamma m$-version CLF-CBF-QP formulation \cite{JANKOVIC2018robust,Han2024Safety,Wang2025Further}. 
This differs from the nominal-input CBF-QP in Section~\ref{subsec:control-affine-qp} only in
the performance part: instead of minimizing the deviation from a prescribed
nominal input, the CLF-CBF-QP uses a relaxed CLF constraint with a vector
slack variable $\delta$. The CBF and auxiliary constraints remain hard and have
the same affine form as in Section~\ref{subsec:control-affine-qp}. With
$V(\mathbf x)=3x_1^2+0.5x_2^2$ and $\alpha(|\mathbf x|)=0.5|\mathbf x|^2$, the controller is defined by
\begin{equation*}
\begin{aligned}
    (\mathbf u,\delta)&\in\arg\min_{\mathbf u,\delta}~
    \frac12\left(|\mathbf u|^2+m|\delta|^2\right)\\
    \text{s.t.}\quad
    &\gamma_f\big(L_fV(\mathbf x)+\alpha(|\mathbf x|)\big)
      +L_gV(\mathbf x)\mathbf u+L_gV(\mathbf x)\delta\le0,\\
    &a_h(\mathbf x)+b_h(\mathbf x)^\top\mathbf u\ge0,\\
    &a_s(\mathbf x)+b_s(\mathbf x)^\top\mathbf u\ge0,
\end{aligned}
\end{equation*}
where $m\ge1$, $\gamma\ge1$, and
$\gamma_f(s)=\gamma s$ for $s\ge0$ and $\gamma_f(s)=s$ for $s<0$.
The compatibility condition in Proposition~\ref{prop:tangential-compatibility}
is explicit in this example. Let $r=|\mathbf x-\mathbf c|$,
$e_r=(\mathbf x-\mathbf c)/r$, and $e_\phi=(-e_{r,2},e_{r,1})^\top$. Then
$b_h=2r e_r$ and, for the selected negative branch, $b_s=-r^{-1}e_\phi$.
Since $\mathbf u_h=0$ satisfies the CBF constraint on $\mathcal C$, choosing
$d=-r e_\phi$ gives $b_h^\top d=0$ and $b_s^\top d=1$. Hence, the CBF and auxiliary constraints are compatible in this case.

In the numerical implementation, we set
$m=5$, $\gamma=2$, $\eta=0.8$, and $h_{\rm gate}=0.12$. The simulations are run
from the initial conditions $(-1,4.5)$, $(0,4.9)$, $(2,4.5)$, $(3,4)$, and
$(2,3)$. The prescribed boundary layer is
$\Sigma_\rho=\{\mathbf x:0\le h(\mathbf x)\le\rho\}$ with $\rho=0.12$. For this
layer, define
\begin{equation}
    \sigma_\rho
    :=
    \inf_{0\le \zeta\le\rho}\sigma(\zeta)
    =
    \frac{h_{\rm gate}^2}{h_{\rm gate}^2+\rho^2},
    \quad
    \eta_\rho:=\eta\sigma_\rho .
    \label{eq:single-effective-eta}
\end{equation}
Thus, whenever $\mathbf x(t)\in\Sigma_\rho$, the hard auxiliary constraint gives
$\dot W(\mathbf x(t),\mathbf u(t))\le-\eta_\rho$. With the above parameters,
$\sigma_\rho=0.5$, $\eta_\rho=0.4$, and the residence bound associated with
$M_U=\pi$ is $2M_U/\eta_\rho\approx 15.71$ s. The principal branch of
$\operatorname{atan2}$ is discontinuous on
$\mathcal B:=\{\mathbf x:x_2=c_2,\ x_1<c_1\}$, so the local-chart certificate is
applied only to residence intervals that do not cross $\mathcal B$; no such
crossing is observed in the reported auxiliary simulations.

Figures~\ref{fig:single-integrator-trajectories} and
\ref{fig:single-integrator-h} compare the baseline and auxiliary-augmented
controllers. For the baseline controller, the upper boundary point
$\mathbf x^*=(0,4.5)$ is an undesired equilibrium of the closed-loop
system: at this point, the CLF descent direction points into the obstacle and is
blocked by the active CBF constraint. This equilibrium acts as a locally attracting undesired equilibrium\footnote{This is the same type of CLF-CBF-QP-induced boundary equilibrium studied in
\cite{reis2020control}, where such undesired equilibria are shown to arise on
the safe-set boundary and can be asymptotically stable under explicit curvature
conditions.}, which explains the long boundary-layer residence observed in
Fig.~\ref{fig:single-integrator-h}(a).
Over the simulated horizon, the longest recorded baseline residence interval is
$18.8$ s. In contrast, the auxiliary-augmented controller induces tangential
motion near the obstacle boundary and reduces the maximum observed
boundary-layer residence time to $4.05$ s, which is below the certified bound.

\begin{figure}
    \centering
    \includegraphics[width=\linewidth]{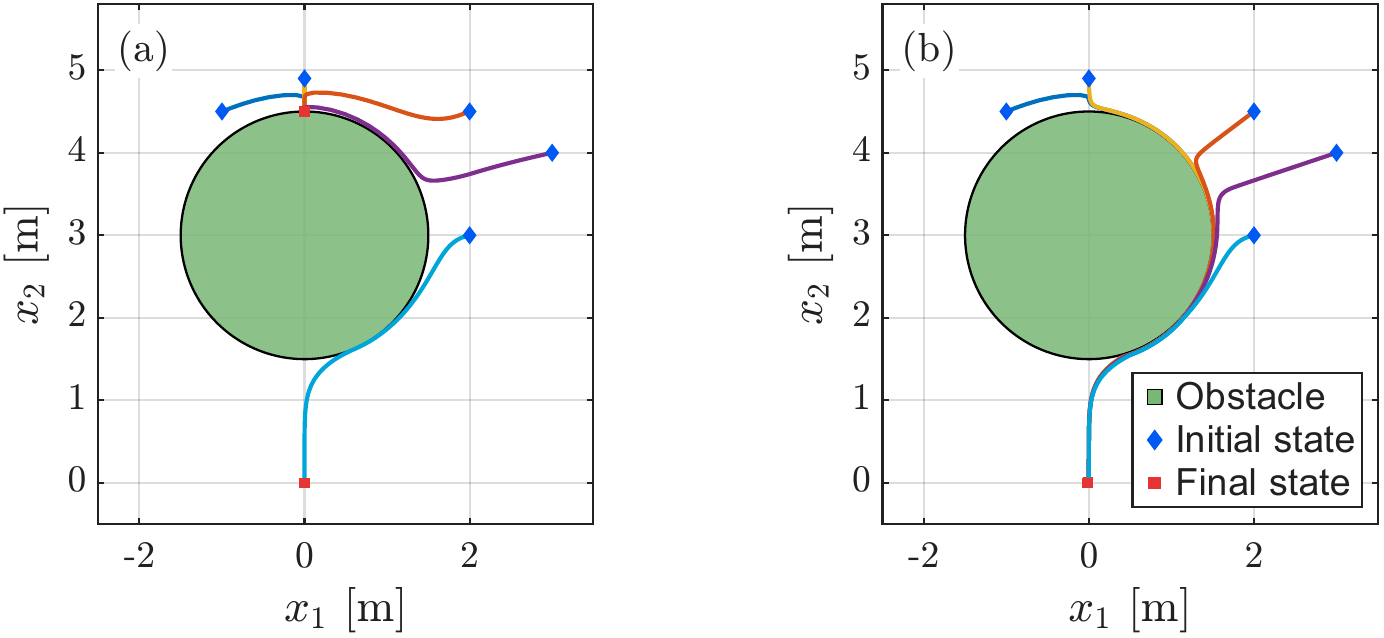}
    \caption{
Planar trajectories for the single-integrator obstacle-avoidance example.
Panel (a) shows the baseline CLF-CBF-QP controller, and panel (b)
shows the auxiliary-augmented CLF-CBF-QP controller. 
}
\label{fig:single-integrator-trajectories}
\end{figure}

\begin{figure}
    \centering
    \includegraphics[width=\linewidth]{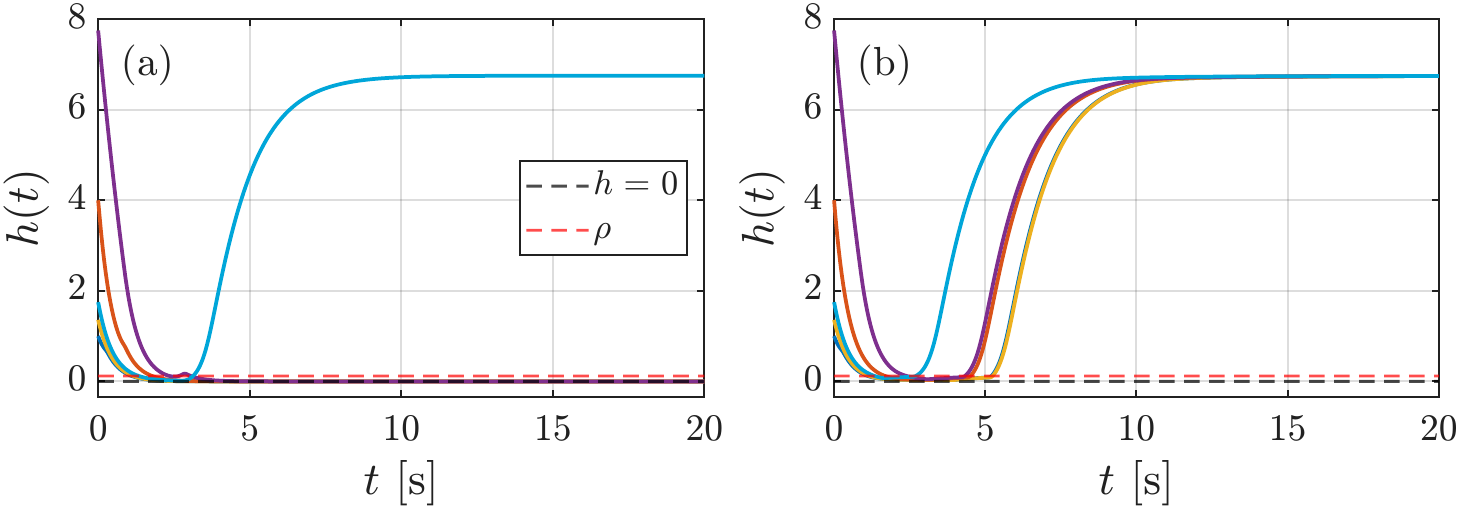}
\caption{
Barrier function values for the single-integrator example. Boundary-layer
residence occurs when $0\le h(t)\le\rho$. Panels (a) and (b) correspond to the
baseline and auxiliary-augmented $\gamma m$ CLF-CBF-QP controllers,
respectively.
}
\label{fig:single-integrator-h}
\end{figure}

\subsection{Double-Integrator System}

Consider the planar double-integrator system
\begin{equation}
    \dot{\mathbf x}=\mathbf v,\qquad \dot{\mathbf v}=\mathbf u,
    \label{eq:double-integrator}
\end{equation}
where $\mathbf x,\mathbf v,\mathbf u\in\mathbb R^2$. We use the same circular
obstacle as in \eqref{eq:safe_set_circle}, with
$h(\mathbf x)=|\mathbf x-\mathbf c|^2-R^2$. Since $h$ has relative degree two,
define $h_1(\mathbf x,\mathbf v):=\dot h(\mathbf x,\mathbf v)+\alpha_1h(\mathbf x)$,
where $\dot h=2(\mathbf x-\mathbf c)^\top\mathbf v$. The HOCBF condition is $\dot h_1(\mathbf x,\mathbf v,\mathbf u)\ge -\alpha_2h_1(\mathbf x,\mathbf v)$.
Let $\mathbf p=\mathbf x-\mathbf c$. Since
$\dot h_1=2|\mathbf v|^2+2\mathbf p^\top\mathbf u+\alpha_1\dot h$, this
constraint is affine in $\mathbf u$: $a_{h_1}(\mathbf x,\mathbf v)+b_{h_1}(\mathbf x)^\top\mathbf u\ge0$,
where $a_{h_1}=2|\mathbf v|^2+\alpha_1\dot h+\alpha_2h_1$ and
$b_{h_1}=2\mathbf p$. If
$h(\mathbf x_0)\ge0$ and $h_1(\mathbf x_0,\mathbf v_0)\ge0$, then the
position-level safe set is forward invariant.
Moreover, along trajectories satisfying $h_1(t)\ge0$, one has
$\dot h(t)\ge -\alpha_1 h(t)$. Hence, the position-level function $h$ satisfies
the non-strict barrier inequality required for the residence certificate.

For the auxiliary function, we use the velocity heading
$W(\mathbf v)=\operatorname{atan2}(v_2,v_1)$ on a selected local branch. Away
from $\mathbf v=0$, $\dot W(\mathbf v,\mathbf u)={(v_1u_2-v_2u_1)}/{|\mathbf v|^2}$.
We do not use the singular coefficient $1/|\mathbf v|^2$ in the QP. Instead,
inside the boundary layer
$\Sigma_\rho=\{(\mathbf x,\mathbf v):0\le h(\mathbf x)\le\rho\}$, we impose the
scaled affine condition $v_1u_2-v_2u_1\ge \eta_0|\mathbf v|^2$.
Equivalently, $a_s(\mathbf v)+b_s(\mathbf v)^\top\mathbf u\ge0$, where
$a_s(\mathbf v)=-\eta_0|\mathbf v|^2$ and
$b_s(\mathbf v)=(-v_2,v_1)^\top$. At $\mathbf v=0$, this scaled constraint is
inactive.

The nominal control is
$\mathbf u_{\rm nom}=-k_p\mathbf x-k_v\mathbf v$.
The auxiliary-augmented safe filter is
\begin{equation*}
\begin{aligned}
    \mathbf u\in\, &\arg\min_{\mathbf u}~
    \frac12|\mathbf u-\mathbf u_{\rm nom}|^2\\
    \text{s.t.}\quad
    &a_{h_1}(\mathbf x,\mathbf v)+b_{h_1}(\mathbf x)^\top\mathbf u\ge0,\\
    &a_s(\mathbf v)+b_s(\mathbf v)^\top\mathbf u\ge0,
      \quad \text{if } 0\le h(\mathbf x)\le\rho .
\end{aligned}
\end{equation*}
Outside $\Sigma_\rho$, the auxiliary constraint is inactive. The baseline HOCBF filter is obtained by omitting the auxiliary constraint. The compatibility condition in Proposition~\ref{prop:tangential-compatibility}
can be checked explicitly for the scaled auxiliary constraint. Let
$\mathbf p^\perp=(-p_2,p_1)^\top$. In the boundary layer, $\mathbf p\neq0$, and
the HOCBF half-space is nonempty for unconstrained inputs. Moreover,
$b_{h_1}=2\mathbf p$ and $b_s=(-v_2,v_1)^\top$. If
$\mathbf p^\top\mathbf v\neq0$, choose
$d=\operatorname{sgn}(\mathbf p^\top\mathbf v)\mathbf p^\perp/|\mathbf p^\top\mathbf v|$.
Then $b_{h_1}^\top d=0$ and $b_s^\top d=1$. Hence the HOCBF and auxiliary
constraints are compatible for unconstrained inputs at such points. When
$\mathbf p^\top\mathbf v=0$, this sufficient CBF-neutral direction condition
does not apply, and feasibility must be checked pointwise.

The residence certificate is local. On any uninterrupted residence interval in
$\Sigma_\rho$ where the selected velocity-heading branch is not crossed and
$|\mathbf v(t)|$ stays bounded away from zero, the auxiliary constraint gives
$\dot W\ge\eta_0$. Hence Proposition~\ref{prop:local-chart} applies with
$M_U=\pi$, and each certified uninterrupted residence interval satisfies $t_b-t_a\le {2M_U}/{\eta_0}$.
This certificate is conditional on $|\mathbf v(t)|$ staying away from zero and
does not apply to intervals that cross the selected branch cut.

\begin{figure}
    \centering
    \includegraphics[width=\linewidth]{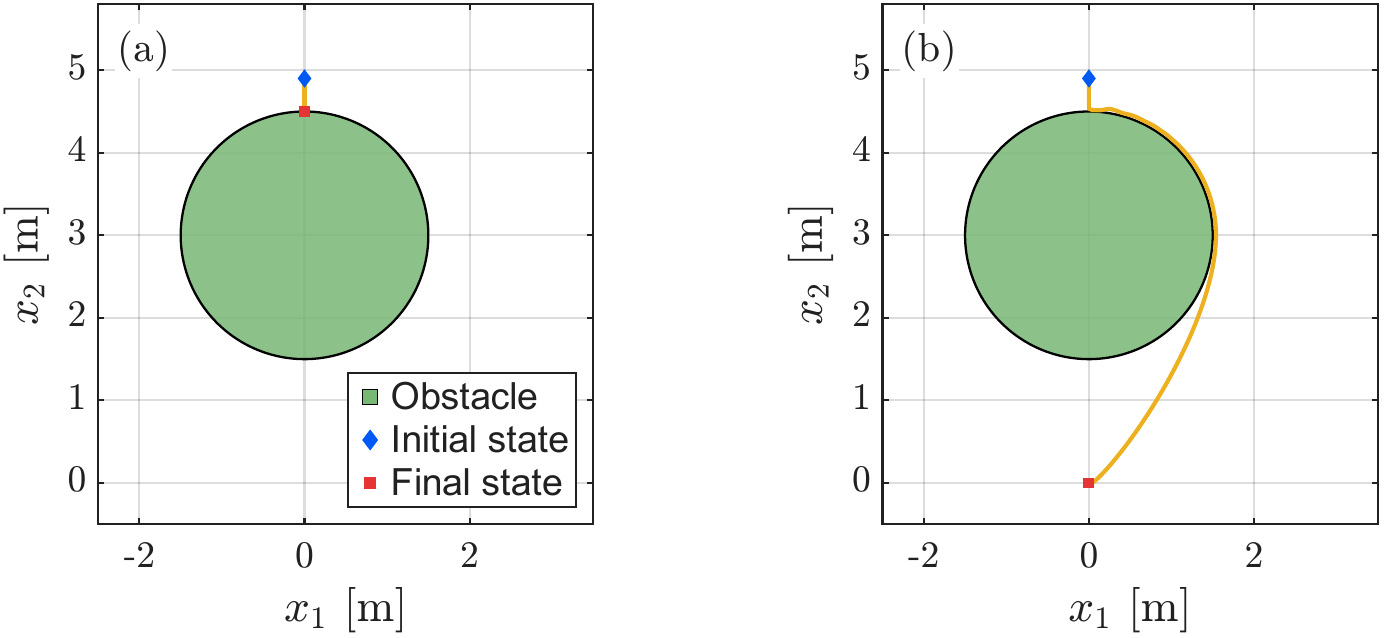}
    \caption{
Planar trajectories for the double-integrator obstacle-avoidance example.
Panel (a) shows the baseline HOCBF-QP safety filter, and panel (b) shows the
boundary-layer auxiliary-augmented HOCBF-QP safety filter.
}
\label{fig:double-integrator-trajectories}
\end{figure}

\begin{figure}
    \centering
    \includegraphics[width=\linewidth]{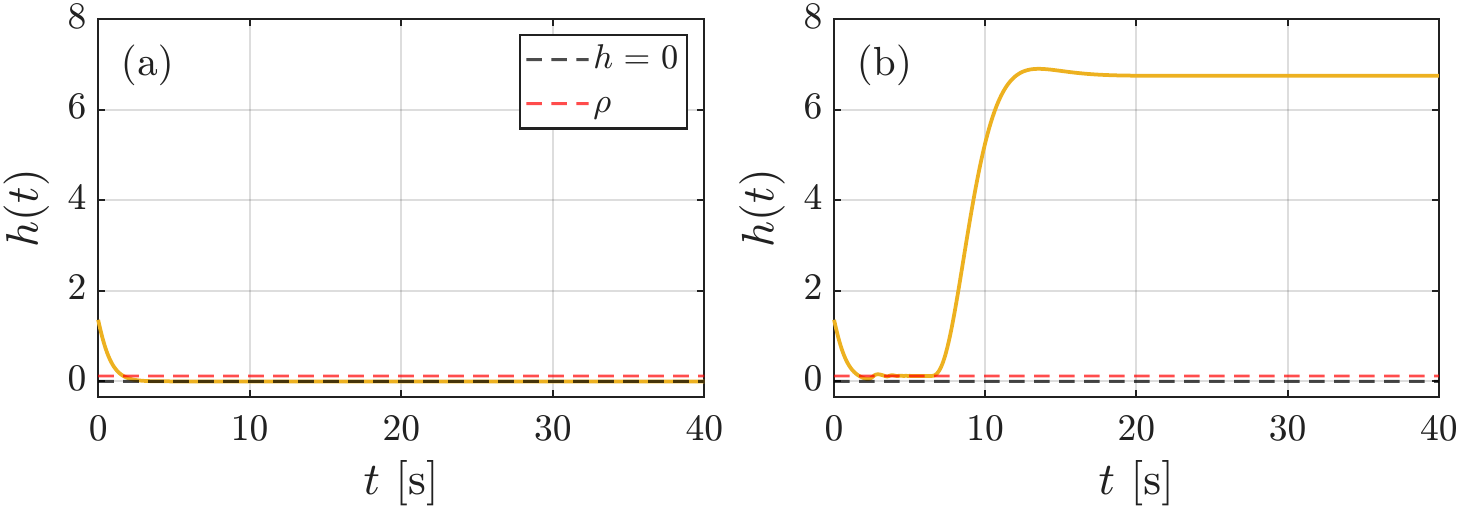}
    \caption{
Barrier function values for the double-integrator example. Boundary-layer
residence occurs when $0\le h(t)\le\rho$. Panels (a) and (b) correspond to the
baseline and boundary-layer auxiliary-augmented HOCBF-QP safety filters,
respectively.
}
\label{fig:double-integrator-h}
\end{figure}

In the simulation, we set $\alpha_1=\alpha_2=2$, $k_p=0.5$, $k_v=1.2$, $\rho=0.12$, $\eta_0=0.8$, and $T=40$ s. The initial
condition is $\mathbf x_0=(0,4.9)$ and $\mathbf v_0=(0,-0.4)$. 
For the baseline controller, the trajectory enters the boundary layer at
approximately $1.67$ s and remains there until the end of the
$40$ s simulation horizon\footnote{This is the same type of symmetry-induced
deadlock behavior discussed in \cite{grover2020does}, where the
nominal stabilizing direction and the active safety constraint can balance on the safe-set boundary and prevent tangential progress.}.
For the auxiliary-augmented filter, all sampled QPs are feasible, no velocity-heading branch cut crossing is observed, and all recorded residence intervals are
certified under the selected local chart. The maximum certified uninterrupted
residence interval is $0.94$ s, which is below the bound
$2\pi/\eta_0\approx7.85$ s. The minimum speed over the certified residence
intervals is $0.063$, above the threshold $0.05$ used in the certificate check.

\subsection{Nonholonomic Unicycle}

Consider the nonholonomic unicycle
\begin{equation}
    \dot x=v\cos\theta,\quad
    \dot y=v\sin\theta,\quad
    \dot\theta=\omega,
    \label{eq:unicycle}
\end{equation}
where $\mathbf q=(x,y,\theta)\in\mathbb{R}^3$ and $\mathbf u=(v,\omega)$. The obstacle is the
same disk used above, with center $\mathbf c=[0,3]^\top$ and radius $R=1.5$.
Let $\mathbf z=(x,y)^\top$, $\mathbf p=\mathbf z-\mathbf c$, and
$h(\mathbf z)=|\mathbf p|^2-R^2$. With
$e_\theta=(\cos\theta,\sin\theta)^\top$, one has
$\dot h=2\mathbf p^\top e_\theta v$. We use $\alpha_h(h)=2h$, so the CBF
constraint is $a_h+b_h^\top\mathbf u\ge0$, where
$a_h=2h$ and $b_h=(2\mathbf p^\top e_\theta,0)^\top$.

The nominal controller is the polar stabilizer \cite{Aicardi1995closed} toward
$(0,0,0)$. Let $r_g=|\mathbf z|$,
$\psi_g=\operatorname{atan2}(-y,-x)$, and
$\alpha=\psi_g-\theta$ on the principal branch. We set
\begin{align}
    v_{\rm nom}
    &=
    k_\rho r_g\cos\alpha,\\
    \omega_{\rm nom}
    &=
    k_\alpha\alpha
    +
    k_\rho\operatorname{sinc}(2\alpha)(\alpha+\lambda\psi_g),
    \label{eq:unicycle-nominal}
\end{align}
where $\operatorname{sinc}(s)=\sin(s)/s$ for $s\neq0$ and
$\operatorname{sinc}(0)=1$.

For the auxiliary function, let
$\phi=\operatorname{atan2}(p_2,p_1)$ and
$\chi=\theta-\phi$ on a selected local branch. Define
$W(\mathbf q)=\operatorname{atan}(k_W\chi)$. For $|\mathbf p|>0$,
\begin{equation}
    \dot W
    =
    c_\chi\left(-\frac{\sin\chi}{|\mathbf p|}v+\omega\right),
    \quad
    c_\chi=\frac{k_W}{1+(k_W\chi)^2}.
    \label{eq:unicycle-Wdot}
\end{equation}
Inside the boundary layer $0\le h\le\rho$, we impose the hard auxiliary
condition $\dot W\ge\eta_0$. Thus $a_s+b_s^\top\mathbf u\ge0$, where
$a_s=-\eta_0$ and
$b_s=c_\chi(-\sin\chi/|\mathbf p|,1)^\top$. Outside the boundary layer, this
auxiliary constraint is inactive.

The baseline and auxiliary-augmented safety filters are both implemented as
\begin{equation*}
\begin{aligned}
    \mathbf u\in\;&\arg\min_{\mathbf u}~
    \frac12|\mathbf u-\mathbf u_{\rm nom}|^2\\
    \text{s.t.}\quad
    &a_h+b_h^\top\mathbf u\ge0,\\
    &a_s+b_s^\top\mathbf u\ge0,
      \quad \text{if } 0\le h\le\rho .
\end{aligned}
\end{equation*}
For the baseline controller, the last constraint is removed.

The compatibility condition in Proposition~\ref{prop:tangential-compatibility}
is immediate for this scaled auxiliary constraint. In the boundary layer,
$\mathbf p\neq0$ and $c_\chi>0$. Choosing
$d=(0,c_\chi^{-1})^\top$ gives $b_h^\top d=0$ and $b_s^\top d=1$. Hence the
CBF and auxiliary half-spaces are compatible for unconstrained inputs.
The residence certificate is local. On any uninterrupted residence interval in
$\Sigma_\rho$ that remains in the selected angular chart, the auxiliary
constraint gives $\dot W\ge\eta_0$. Since
$|W(\mathbf q)|=|\arctan(k_W\chi)|\le \pi/2$, Proposition~\ref{prop:local-chart}
applies with $M_U=\pi/2$. Hence, each certified uninterrupted residence interval
satisfies $t_b-t_a\le {2M_U}/{\eta_0}={\pi}/{\eta_0}$.
This certificate is conditional on remaining in the selected angular chart and on the pointwise feasibility of the auxiliary constraint. 

In the simulation, we used $k_\rho=0.8$, $k_\alpha=2.5$, $\lambda=0.8$,
$k_W=0.8$, $\eta_0=0.4$, $\rho=0.12$, and $T=15\,{\rm s}$ with time step
$0.01\,{\rm s}$. The initial states were $(0,4.9,-\pi/2)$,
$(1.6,4.5,-0.8\pi)$, $(2.5,3.8,\pi)$, and $(1.5,2.2,-2.2)$. The corresponding
trajectories and barrier values are shown in Figs.~\ref{fig:unicycle-trajectory}
and~\ref{fig:unicycle-h}. Both controllers keep the trajectories safe and all
sampled QPs are feasible. However, the baseline CBF-QP exhibits prolonged motion along the obstacle
boundary. The auxiliary-augmented CBF-QP steers the trajectories around the
obstacle and then toward the origin. In the auxiliary simulations, all sampled
QPs are feasible, no selected angular branch cut crossing is observed, and the
maximum sampled residence time is reduced to $0.97\,{\rm s}$, which is below the
local-chart bound $\pi/\eta_0\approx 7.85\,{\rm s}$, while safety is maintained
with $h_{\min}=0.023$.

\begin{figure}
    \centering
    \includegraphics[width=0.95\linewidth]{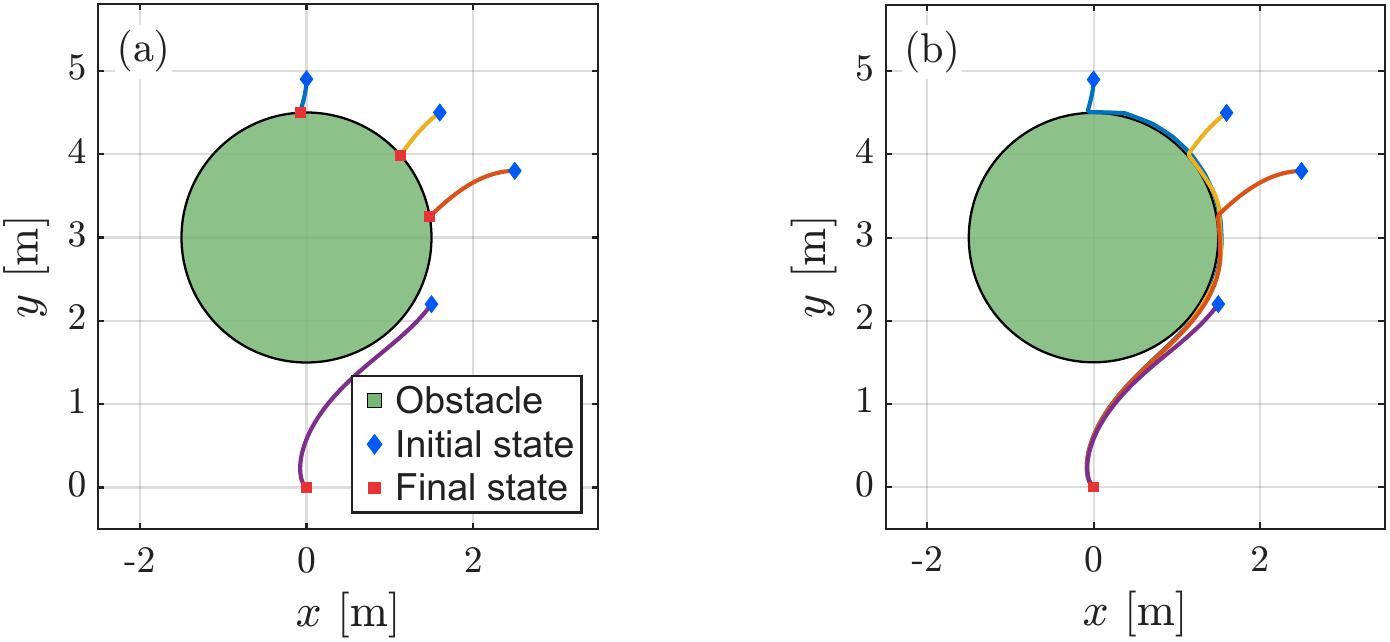}
    \caption{Planar trajectories for the unicycle obstacle-avoidance example. Panel (a) shows the baseline CBF-QP controller, and panel (b) shows the auxiliary-augmented CBF-QP controller.}
    \label{fig:unicycle-trajectory}
\end{figure}

\begin{figure}
    \centering
    \includegraphics[width=0.95\linewidth]{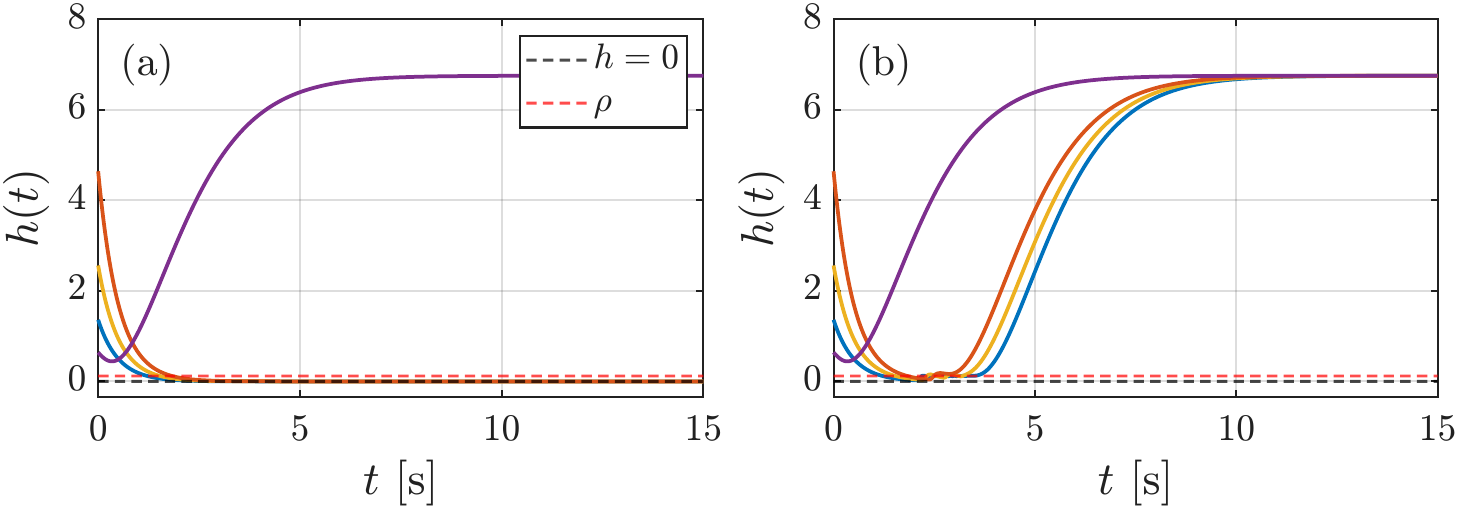}
    \caption{Barrier function values for the unicycle example. Boundary-layer residence occurs when $0\le h(t)\le\rho$. Panels (a) and (b) correspond to the baseline and auxiliary-augmented CBF-QP controllers, respectively.}
    \label{fig:unicycle-h}
\end{figure}

\section{Conclusion} \label{sec:conclusion}

This paper developed a finite continuous boundary-layer residence certificate
for systems satisfying non-strict CBF conditions. The certificate preserves the
standard forward-invariance guarantee and adds a bounded auxiliary function
whose one-sided derivative bound in a prescribed boundary layer gives an
explicit upper bound on every uninterrupted residence interval. Thus, the result
addresses trajectory-level near-boundary residence without replacing the
non-strict CBF condition by a strict barrier condition or requiring a
classification of boundary equilibria. 
For control-affine systems, the auxiliary condition was implemented as an
additional affine constraint in a CBF-QP. A tangential-input compatibility
condition was given to ensure simultaneous feasibility with the hard CBF
constraint for unconstrained inputs. Local-chart versions handle angular or
multi-valued auxiliary functions, such as $\operatorname{atan2}$, provided the
certified residence interval remains in the selected chart and avoids branch
cuts or singularities. Single-integrator, double-integrator, and unicycle
examples illustrated the resulting radial--tangential construction and the
associated residence-time bounds.
Future work will address
cumulative residence bounds, robustness margins, input-constrained feasibility,
chart-transition arguments, and systematic construction of auxiliary functions.

\printcredits




\bibliographystyle{model1-num-names}
\bibliography{v1-main}

\end{document}